\documentclass[preprint]{imsart}
\usepackage[a4paper, total={6in, 8in}]{geometry}
\RequirePackage{amsthm,amsmath,amsfonts,amssymb}
\RequirePackage[authoryear,numbers,sort&compress]{natbib}

\RequirePackage[colorlinks,citecolor=blue,urlcolor=blue]{hyperref}
\RequirePackage{graphicx}
\usepackage{comment}
\usepackage{multirow}
\usepackage{float}
\usepackage{subcaption}
\usepackage{thmtools, thm-restate}
\startlocaldefs
\theoremstyle{plain}

\newtheorem{theorem}{Theorem}[section]
\newtheorem{definition}[theorem]{Definition}

\newtheorem{proposition}{Proposition}[section]
\newtheorem{corollary}[proposition]{Corollary}
\newtheorem{rem}{Remark}

\theoremstyle{remark}

\endlocaldefs

\begin{document}

\begin{frontmatter}
\title{Sensitivity Analysis for False Discovery Rate Estimation with Published p-Values}
\runtitle{FDR Estimation with Published p-Values}

\begin{aug}
\author[A]{\fnms{Tianyu}~\snm{Cao}\ead[label=e1]{tcao@uwyo.edu}},
\author[B]{\fnms{Sangyoon}~\snm{Yi}\ead[label=e2]{sayi@okstate.edu}}
\and
\author[B]{\fnms{Joshua}~\snm{Habiger}\ead[label=e3]{jhabige@okstate.edu}}
\address[A]{Department of Zoology and Physiology,
University of Wyoming\printead[presep={,\ }]{e1}}

\address[B]{Department of Statistics,
Oklahoma State University \printead[presep={,\ }]{e2,e3}}
\end{aug}

\begin{abstract}
There is recent interest in estimating the false discovery rate (FDR) with published p-values. However, there is little formal research that addresses the manner and extent to which the presumed selection, or publication, bias model impacts the bias and variance of FDR estimators. This manuscript provides general and closed-form expressions for the bias and variance of an established FDR estimator when the publication bias model ($p<0.05$) may or may not be correct.  Expressions reveal that FDR estimates could be conservative or liberal, depending on how well a $p<0.05$ publication rule approximates the true selection mechanism. Analysis of a well-studied large-scale replication project in psychology, where selection model parameters are estimable, suggests that bias expressions are accurate in practice. Another well-studied collection of p-values mined from medical journal abstracts is used to illustrate how provided closed-form expressions may facilitate a simple sensitivity analysis when the goal is FDR estimation using selected p-values with unknown selection mechanism. 
 \end{abstract}



\end{frontmatter}
\baselineskip=20pt
\section{Introduction}
\cite{Sterling1959, Rosenthal1979, Soric1989, Ioannidis2005} argued that bias occurring from publication policies that favor statistical significance, called publication bias, could lead to a high false discovery rate (FDR) among published results.  The basic argument is understood through the lens of Bayes' theorem: If most null hypotheses submitted for publication are true (i.e., the prior null probability is high) and a result is published only if $p<0.05$, then the posterior null probability may still be high even if $p<0.05$. Consequently, the proportion of false discoveries among published discoveries may still be high. Of course, if the prior null probability is low, then similar arguments imply that the FDR among published results is low. \cite{Ioannidis2005,Higginson2016,Campbell2019}, among many others, argue that publication incentives, limited resources, and $p<\alpha$ publication policies naturally lead to a high prior null probability.  However, in practice, the prior null probability is rarely known, and mechanisms leading to the publication and observation of a $p$-value are more nuanced than ``$p<0.05$'' (cf. \cite{Benjamini2020}). 

This has led to an interest in FDR estimation with published, or more generally, ``selected'', $p$-values. For example, \cite{Jager2014, Schimmack2023} selected 5000+ statistically significant $p$-values from abstracts of top medical journals between 2000 and 2010 and estimated $\widehat{\text{FDR}} = 0.14 (\pm 0.01$). \cite{Schneck2023} estimated the FDR using p-values and test statistics mined from 35,515 psychology papers published between 1975 and 2017, and reported $\widehat{\text{FDR}} = 0.177$. \cite{Johnson2017} analyzed \cite{OSC2015} replication study data, which contains a selection of 73 published $p$-values in psychology and their corresponding replicated $p$-values. Approximately, thirty-four out of 70 original statistically significant discoveries were estimated to be false discoveries, which yields a false discovery proportion of about 0.49.

FDR estimators that utilize only published $p$-values are sometimes criticized because they must assume some selection model, and results may or may not be robust with respect to selection model misspecification. For example, \cite{Benjamini2014b, Goodman2014, Ioannidis2014, Gelman2014b} all noted that the FDR estimator in \cite{Jager2014}, which assumed that the publication probability is constant for $p<0.05$ and 0 otherwise, did not account for the fact that only p-values from the abstracts were mined.  Furthermore, \cite{Benjamini2014b} noted that selecting $p$-values using ``$p\leq 0.05$'' rather than ``$p<0.05$'' caused the FDR estimate to shift from 0.14 to 0.20. The selection model assumed in \cite{Johnson2017} stipulates that $p$-values less than 0.052 are published with some probability, say $c$, and that $p$-values greater than 0.052 are published with probability $c\times 0.006$.  Although more flexible than the standard ``$p<0.05$'' rule, this stipulation may not still be reasonable. Indeed it stipulates that $p = 0.051$ is 167 times more likely to be published than $p=0.053$.  However, perhaps this assumption has minimal impact on the FDR estimate; perhaps it does not. In general, our ability to assess, diagnose and address the health of our scientific ecosystem with published p-values relies, at least in part, on the answer to these types of questions. This manuscript is aimed at providing such answers. Specifically, we study the effect of publication bias model misspecification on FDR estimators. 

To facilitate a rigorous study and to reach a broad audience, this manuscript focuses on a simple FDR estimator, which is motivated under a simple publication bias model.  The estimator, formally introduced within this context in \cite{Hung2020}, first adjusts published p-values via $p_i' = p_i/\alpha$ and then applies a well-established method of moments estimator to the adjusted $p$-values via 
\begin{equation}\label{pi0.hat} \widehat{\pi}_0 = \frac{1}{m}\sum_{i=1}^m \frac{1_{(\lambda,1]}(p'_i)}{(1-\lambda)},
\end{equation}
where $\lambda$ is a tuning parameter and $1$ is the indicator function. The p-value adjustment can be motivated using selective inference theory \citep{Fithian2014} if we assume the selection probability is constant when $p$-values are less than $\alpha$ and 0 otherwise. The estimator, introduced in \cite{Schweder1982} as an estimate for the proportion of uniformly distributed $p$-values (or null $p$-values), has been well studied in the FDR and multiple testing literature.  For example, it is known to be positively biased when $p$-values are uniformly distributed under the null hypotheses.  Further, the bias decreases to 0 as the powers of the individual tests increase to 1. Additionally, the bias and variance of the estimator depend upon $\lambda$ and hence can be ``tuned.'' See \cite{Storey2002,Storey2004,  Blanchard2009}. For more recent studies see \cite{Gao2025} and references therein.   

This manuscript provides general expressions for the bias and variance of the estimator under arbitrary selection models. It is shown that the properties of the FDR estimator (e.g., positive bias that decreases in $\lambda$ and power) are retained if the selection model used in $p$-value adjustment is correct, but need not be retained otherwise. Closed-form expressions for the bias and variance of the FDR estimator under some common parametric publication bias models are also derived and explored. For example, we see that the bias can be negative and does not need to be monotone in $\lambda$ when the selection model is misspecified. 

To validate theoretical results in practice, we demonstrate that bias expressions generally account for discrepancies between the proposed FDR estimator, which uses only originally published p-values, and the gold standard estimate in \cite{Johnson2017} that used both original and replication p-values. Thus the proposed estimator, along with bias adjustments, can potentially allow for FDR estimation that is adjusted for publication bias using only selected $p$-values under some selection model. We additionally illustrate a sensativity analysis on p-values mined from the abstracts of medical journals in \cite{Jager2014}.  In particular, expressions allow for FDR estimates to be reported across a range of potential publication bias models. It is anticipated that such an analysis may inform policy makers who aim to diagnose and, if necessary, address replicability issues arising from publication bias. 

\section{Models}

Assume there exists a complete set of $p$-values $\left\{P_i\right\}_{i=1}^{m}$ to test null hypothesis $\{H_{i}\}_{i=1}^m$. Following the usual convention, we let $H_{i}$ be either 0 or 1 if the corresponding null is true or false, respectively. The following p-value mixture model is commonly used in the FDR literature \citep{Genovese2002, Storey2003} and throughout this manuscript:
\begin{definition}[P-value Mixture Model]\label{MM}
\normalfont Let $\{(H_i, P_i)\}_{i=1}^{m}$ be $i.i.d.$ bivariate random vectors from the joint distribution satisfying $H_i\sim\text{Bernoulli}(1-\pi_{0})$ and
\[
\Pr(P_i\leq t|H_i)= (1-H_{i})F_0(t)+H_{i}  F_1(t)=(1-H_{i}) t+H_{i} F_1(t),  
\]
where $\pi_0\in (0,1)$ and $F_{1}(t)\geq t$ for all $t$ with the inequality being strict for some $t$.
\end{definition}
\noindent In Definition \ref{MM}, $P_{i}$ has a uniform distribution under the null hypothesis (i.e., $H_{i}=0$) with the null probability $\pi_0$. The condition on $F_1(t)$ ensures that $p$-values tend to be smaller when the null hypothesis is false.
It is common to assume that $P_i$ is generated from a Beta distribution when $H_i$ is false, which leads to the specific model:  
\begin{equation}\label{betamixture}
F(t)=\pi_0 t +(1-\pi_0) t^{\gamma}\quad \text{for}\quad 0\leq\gamma<1. 
\end{equation}
Such a class of models has also been used to study the FDR under various publication policies \citep{Jager2014, Benjamin2017, Lakens2018, Habiger2022}.    

As discussed in the Introduction, in practice only a portion of the generated $p$-values are published, and perhaps an even smaller fraction is observed in further selection.  For readability, we refer to an observed $p$-value as ``selected'' and denote the selection event by $\delta_i$, where $\delta_i = 1$ if $P_i$ is observed and 0 otherwise.  Formally, a selection probability model is a statistical model for  $\Pr(\delta_i=1|P_i=p)$.  For example, the ``Publish if $P\leq 0.05$'' policy formally assumes $\Pr(\delta_i = 1|P_i=p) = c \cdot 1_{[0,0.05]}(p)$, i.e $P_i$ is selected with probability $c$ if it is less than or equal to 0.05 and not selected otherwise.  Alternatively, $\Pr(\delta_i = 1|P_i=p) = c$ assumes selection is independent of $P_i$.  Such a model may be reasonable in pre-registered studies, where the decision to publish is made prior to $p$-value generation \citep{Nosek2014a}.  The formal definition is provided below.  

\begin{definition}[Selection Probability Model] \label{def:PPM}
\normalfont A selection probability model (SPM) is a non-increasing function $\mu:[0,1]\rightarrow [0,\infty)$ defined by \begin{equation}
     \mu(p)\propto \Pr(\delta=1|P=p) \label{ppm}
\end{equation}
satisfying $E_{F}[\mu(P)]<\infty$ where $E_{F}[\mu(P)]$ denotes the expectation with respect to the distribution $F$ of $P$. 
\end{definition}
\noindent Note that $\mu(p)$ may be a continuous function, a step function or even piecewise continuous. The important assumption is that the selection probability must be non-increasing in $p$ so that, in general, larger $p$-values are not more likely selected. The assumption $E_{F}[\mu(P)]<\infty$ allows for continuous SPMs to be proportional to a density function, which will be useful in calculations.  

Let us more carefully consider some examples before proceeding.  \cite{Jager2014} assumed $\mu(p) = 1_{[0,0.05)}(p)$ in their estimate of the science-wise FDR with $p$-values mined from the abstracts of medical journals while \cite{Benjamini2014b} studied $\mu(p) = 1_{[0,0.05]}(p)$.  \cite{Hung2020} also considered the former model in their calibration of original $p$-values from the \cite{OSC2015} replication study while   \cite{Johnson2017} considered
\begin{equation}\label{johnsonrule}
     \mu(p)= 1_{[0, 0.052)}(p)+ 0.006\times 1_{[0.052, 1]}(p),
\end{equation}
in their analysis of the original $p$-values. \cite{Moss2023} used 
\begin{equation}
    \mu(p)=1_{[0, 0.025)}(p)+0.7\times 1_{[0.025, 0.05)}(p)+0.1\times 1_{[0.05, 1]}(p), \nonumber
\end{equation}
to study p-hacking. On the other hand, \cite{citkowicz2017parsimonious} assumed 
\begin{equation}
    \mu(p)\propto p^{a-1}(1-p)^{b-1},\nonumber
\end{equation}
for $a,b>0$. For a more systematic review of SPMs, we refer to \cite{Jin2015} and \cite{Hedges1984}. 

Under Definitions \ref{MM}-\ref{def:PPM} the probability density function (pdf) for a selected $p$-value is denoted
\begin{equation}\label{pubppro}
f(t|\delta_i=1)=\frac{\mu(t)f(t)}{E_F[\mu(P)]},
\end{equation}
or 
\begin{equation}
f_j(t|\delta_i=1)=\frac{\mu(t)f_j(t)}{E_{F_j}[\mu(P)]},\nonumber
\end{equation}
under $H_i = j$ for $j=0, 1$, where $f(t)$ is the pdf of original p-values before selection, $f_0(t)$ and $f_1(t)$ are the pdfs of original p-values under null and alternative, respectively. By Definition \ref{MM}, $f_0(t|\delta_i = 1)$ can be shown to be the pdf for a Uniform(0,$\alpha$] distribution if $\mu(p) = 1_{[0,\alpha](p)}$. On the other hand, if $\mu(p)=c$, then it can be verified that $f_0(t|\delta_i=1)$ is the pdf of a Uniform(0,1] distribution. Clearly the distribution of $p$-values depends on the SPM $\mu(p)$. See \cite{Fithian2014} or \cite{dear1992approach} for a formal treatment of post-selection densities in the selective inference and publication bias literature, respectively. 

\section{Post-Selection False Discovery Rate}
Before defining the post-selection FDR, it is important to point out that there is active research aimed at defining a flexible \textit{selection} mechanism that provides FDR control when applied to fully observed data (cf. \cite{Lei2018, Benjamini2014, Lei2020, Barber2015}). In general, post-selective inference aims to adjust inference for parameters after a \textit{known} selection mechanism has been employed. For example, see \cite{Benjamini2005, Yekutieli2012, Lee2016, Fithian2014, Bao2024}.  However, in the current study, the selection mechanism is not known and the complete set of $p$-values is not observable. Post-selective inference with an unknown selection mechanism is primarily studied in publication bias literature, with a focus on effect size estimation. See, for example, \cite{Hedges1984,Hedges1992, McShane2016,Moss2023}. However, little attention is given to the FDR.

To define the post-selection FDR, recall that the classical FDR \citep{Benjamini1995} is $E[V/R|R>0]\Pr(R>0)$ where $V = \sum_{i=1}^m (1-H_i)1_{[0,\alpha]}(P_i)$ is the number of false discoveries and $R = \sum_{i=1}^m1_{[0,\alpha]}(P_i)$ is the total number of discoveries. More mathematically tractable versions of the FDR, such as the marginal or positive false discovery rate defined $\text{mFDR} = E[V]/E[R]$ and $\text{pFDR} = E[V/R|R>0]$, are often used to study and develop FDR methods \citep{Storey2002, Genovese2002, Efron2004, Sun2007}. We shall follow suit. In particular, Theorem 1 in \cite{Storey2003} which gives $\text{pFDR} = \Pr(H_i=0|P_i\leq \alpha)$ under the mixture model in Definition \ref{MM}, inspires a mathematical tractable definition of the post-selection FDR presented below.  
\begin{definition}[Post-Selection False Discovery Rate]\label{def:FDR}
\normalfont Let $\{P_i\}_{i=1}^m$ be generated as in Definition \ref{MM} and $\mu$ be an SPM as in Definition \ref{def:PPM}. The post-selection FDR is defined
\begin{equation}
\text{FDR} = E\left[\frac{\sum_{i=1}^{m}(1-H_i)\delta_i 1_{[0,\alpha]}(P_i)}{\sum_{i=1}^{m}\delta_i1_{[0,\alpha]}(P_i)}\bigg\rvert\sum_{i=1}^m\delta_i1_{[0,\alpha]}(P_i)>0\right]. \nonumber
\end{equation}
\end{definition}
\noindent Observe that when $\delta_i=1$ for all $i$, FDR reduces to $\text{pFDR} = E[V/R|R>0]$.  However, when selection occurs so that $\delta_i = 1$ for only some $i$, the FDR in Definition \ref{def:FDR} is the expected proportion of false discoveries among all \textbf{selected} discoveries, where ``selected'' refers to $\delta_i=1$ and false discovery indicates $P_i\leq \alpha$ while $H_i=0$. Note that more accurate description of the FDR in Definition \ref{def:FDR} would be the ``positive post-selection FDR."  However, since other definitions of the post-selection FDR are not considered in this manuscript, we interchangeably use ``FDR'' or ``post-selection FDR'' and omit ``positive'' throughout this paper.        
The Bayesian-FDR connection in Theorem 1 in \cite{Storey2003} can be easily generalized to relate the post-selection FDR above to a posterior probability. The proof follows the approach in \cite{Storey2003} and is included in the Appendix. 
\begin{theorem}\label{bayes}
\normalfont Assume $H_i=0$ is rejected if $P_i\leq \alpha$ for some $\alpha \in (0,1)$ and let $\delta_i = 1$ denote the event that $P_i$ is selected. Under models in Definitions \ref{MM}-\ref{def:PPM} and the FDR in Definition  \ref{def:FDR},
\begin{align}
\text{FDR}=\Pr(H_i=0|\delta_i=1, P_i\leq\alpha)\nonumber.
\end{align} 
\end{theorem}

The post-selection FDR estimator considered here and formalized in \cite{Hung2020} is the null proportion estimator in \eqref{pi0.hat} but is presented here for formality.
\begin{definition}[Post-Selection FDR Estimator]\label{hungmme}
\normalfont Let $\{p_{i}\}_{i=1}^{n'}$ be the set of selected $p$-values that are less than or equal to $\alpha$ for some $\alpha\in(0,1)$, and $\lambda$ be a fixed tuning parameter in $(0,1)$. Define   
\begin{equation}   \widehat{\text{FDR}}=\sum_{i=1}^{n'}\frac{1_{(\lambda,1]}(p_i')}{n'(1-\lambda)},\nonumber
\end{equation}
where $p_{i}'=p_{i}/\alpha$ is the adjusted $p$-value.
\end{definition}
\noindent A subtle but important point is that more than $n'$ many $p$-values may be selected by the SPM; here $n'$ is the number of selected $p$-values that are also less than or equal to $\alpha$. As in the  \cite{Storey2004}, it is also assumed that at least one selected $p$-value is ``discovered as significant'', i.e. less than $\alpha$, since otherwise FDR estimation is not of interest. 

\section{Bias and Variance Expressions}

Both the post-selection density \eqref{pubppro} and Theorem \ref{bayes} facilitate expressions for the bias and variance of the estimator in Definition \ref{hungmme} under an arbitrary SPM. Note that the bias of the estimator in Definition \ref{hungmme} is defined to be $\mathrm{Bias}(\widehat{\text{FDR}}) = E[\widehat{\text{FDR}}] - \text{FDR}$. In the following Theorem, we provide an expression of the bias in terms of SPM with the proof given in the Appendix.

\begin{theorem}\label{biasMME}
\normalfont Under models in Definitions \ref{MM}-\ref{def:PPM}, we have  
\begin{equation}\label{df:biasMME}
\text{Bias}(\widehat{\text{FDR}})=\frac{(\alpha-\lambda)^{-1}\alpha E_F[1_{[\lambda,\alpha]}(P)\mu(P)]-\pi_0E_{F_0}[1_{[0,\alpha]}(P)\mu(P)]}{E_{F}[1_{[0,\alpha]}(P)\mu(P)]}, 
\end{equation}
for fixed $\lambda\in[0,\alpha]$ and $\alpha \in (0,1]$.
\end{theorem}

To better understand the expression, suppose $p$-values follow the mixture model \eqref{betamixture} and that the ``$p\leq \alpha$'' selection mechanism is employed, i.e. $\mu(p)\propto 1_{[0, \alpha]}(p)$. This SPM and equation \eqref{df:biasMME} gives
\begin{align}
    \mathrm{\text{Bias}}(\widehat{\text{FDR}})&=\frac{\alpha}{\alpha-\lambda}\frac{\pi_0(\alpha-\lambda) + (1-\pi_0)(\alpha^\gamma-\lambda^\gamma)}{\pi_0\alpha + (1-\pi_0)\alpha^\gamma}-\frac{\pi_0 \alpha}{\pi_0 \alpha + (1-\pi_0)\alpha^\gamma} \label{eq:biassimple} \\
    &=\frac{\alpha}{\alpha-\lambda}\frac{(1-\pi_0)(\alpha^\gamma-\lambda^\gamma)}{\pi_0\alpha + (1-\pi_0)\alpha^\gamma}>0, \nonumber
\end{align}
where the last inequality is satisfied since  $\alpha^\gamma - \lambda^\gamma>0$ for $\gamma\in(0,1)$, $\alpha>\lambda$ and $\pi_0<1$. Further, noting that the power is defined as $\Pr(P_i\leq \alpha|H_i=1)= F_1(\alpha) = \alpha^\gamma$, we recover the least favorable configuration for the FDR \citep{Finner2009, Habiger2017a} by taking $\gamma = 0$, i.e, $p$-values vary according to a Dirac $\delta$ distribution with point mass at 0 when the null hypothesis is false and are uniformly distributed otherwise.  Observe $\text{Bias}(\widehat{\text{FDR}}) =0$ when $\gamma=0$. In summary, the estimator is unbiased when the power is 1 and is positively biased otherwise. These results may be expected since the $p$-value calibration model and selection model agree in these calculations. 

To explore settings when the calibration model needs not agree with the selection model, consider the step function SPM defined 
\begin{equation}\label{cs}
\mu(p)=1_{[0, \frac{\alpha}{2})}(p)+\rho 1_{[\frac{\alpha}{2}, \alpha]}(p), 
\end{equation}
where $\rho\in (0, 1]$. Observe that when $\rho=1$ the $p\leq \alpha$ rule is recovered.  Otherwise the $p$-values between $\alpha/2$ and $\alpha$ are less likely selected selected than $p$-values less than $\alpha/2$.  \cite{Moss2023} and \cite{moss2022infinite} suggest that this SPM may account for ``two-tailed'' vs ``one-tailed'' p-value reporting in practice. Using \eqref{df:biasMME}, we can verify that the bias of the $\widehat{\text{FDR}}$ under the SPM \eqref{cs} is given as
\begin{equation}
\text{Bias}(\widehat{\text{FDR}}) = \begin{cases} 
G_1(\alpha,\lambda,\pi_0,\gamma,\rho)/H(\alpha,\pi_0,\gamma,\rho) & \text{if } \lambda \geq \alpha/2, \\
G_2(\alpha,\lambda,\pi_0,\gamma,\rho)/H(\alpha,\pi_0,\gamma,\rho) & \text{if }\lambda < \alpha/2,
\end{cases}\label{cstep}
\end{equation}
where 
\[
\begin{aligned}
G_1(\alpha,\lambda,\pi_0,\gamma,\rho)=&\alpha[\pi_0\rho + (1-\pi_0)\rho(\alpha^\gamma-\lambda^\gamma)/(\alpha-\lambda)-\pi_0(1+\rho)/2], \\    
G_2(\alpha,\lambda,\pi_0,\gamma,\rho)=&\alpha\bigg[ \frac{\pi_0}{\alpha - \lambda} \left\{ (\rho+1)\frac{\alpha}{2}-\lambda \right\} +\frac{(1-\pi_0)}{\alpha - \lambda} \left\{ \rho\left(\alpha^\gamma-\frac{\alpha}{2}^\gamma\right) + \left( \frac{\alpha}{2} \right)^\gamma-\lambda^\gamma \right\}  \\
&-\pi_0(1+\rho)/2\bigg],\\
H(\alpha,\pi_0,\gamma,\rho)=& \frac{ \alpha\pi_0(1+\rho) }{2}+(1-\pi_0) \left[ \left(\frac{\alpha}{2}\right)^\gamma+\rho \left\{ \alpha^\gamma-\left(\frac{\alpha}{2}\right)^\gamma \right\} \right].
\end{aligned}
\]

The expression above readily allows for the relationship between the bias and SPM to be explored.  In particular, Figure \ref{stepSPM} plots the bias in \eqref{cstep} over $\rho$ and power for $\lambda=0.01, 0.025, 0.045$ and $\pi_{0}=0.3, 0.5, 0.8$ when $\alpha=0.05$. Observe the bias can be positive or negative. In particular, for small $\rho$ and $\lambda = \alpha/2$ or $\lambda>\alpha/2$, the bias is negative, especially when prior null probability $\pi_0$ is large. This is expected since the selection model leads to fewer $p$-values between $\alpha/2$ and $\alpha$ being observed, and the adjusted $p$-value $p_i' = p_i/\alpha$ does not take this selection into consideration. For $\lambda < \alpha/2$, the bias may still be positive but is decreasing in $\rho$. When $\pi_0$ is small and the FDR is small, some negative bias may exist, but is minimal. This may be expected since little negative bias is possible when the FDR is already near 0.  We also observe that when the presumed selection model is correct ($\rho = 1$) and $p$-values are generated according to a least favorable distribution ($\gamma = 0$) then the bias is 0 as expected.

\newpage
\begin{figure}[H]
    \includegraphics[width=1.0\textwidth, height = 15cm]{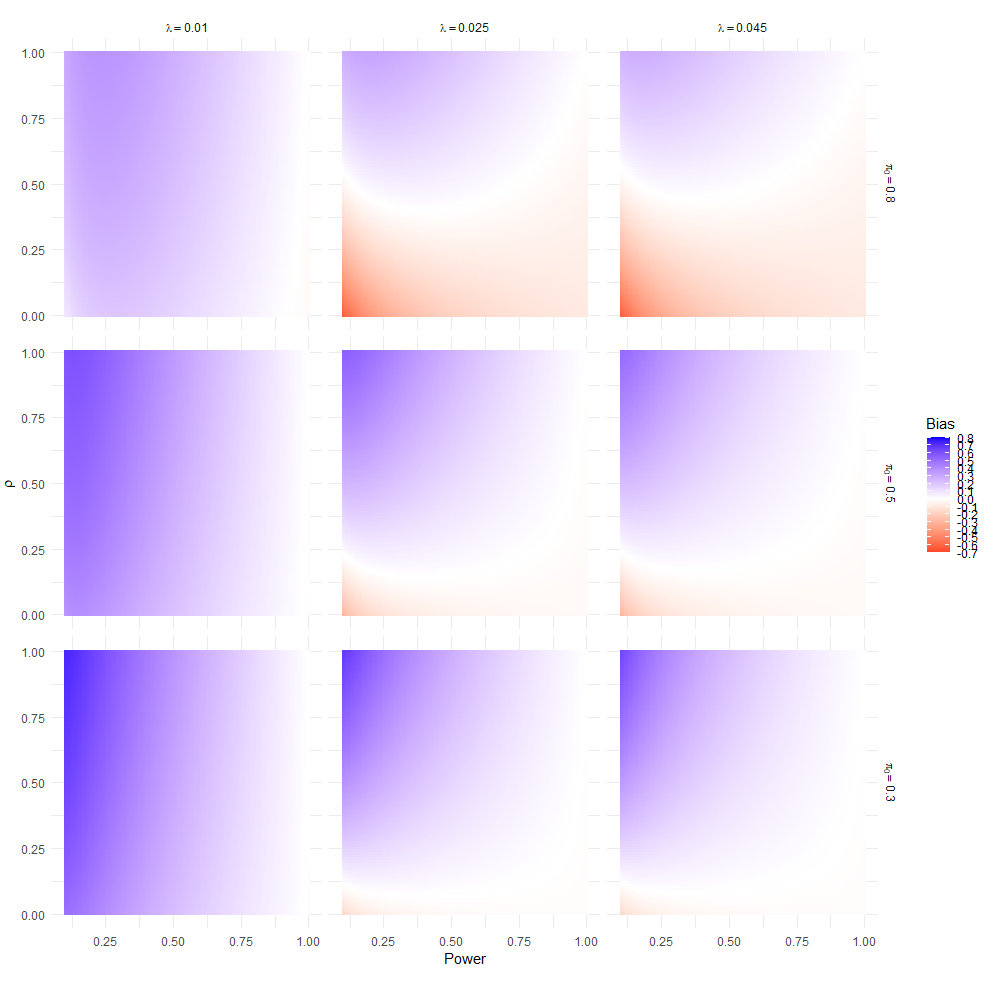}
    \caption{Bias in \eqref{cstep} over $\rho$ and power (corresponding to $\gamma$) at $\lambda = 0.01,0.025,0.045$ and $\pi_{0}=0.3,0.5,0.8$ when $\alpha=0.05$}
    \label{stepSPM}
\end{figure}

\newpage
Next, let us consider an absolutely continuous SPM 
\begin{equation}\label{spbeta}
    \mu(p)= (1-p)^{\eta-1},
\end{equation}
where $\eta\geq 1$. The fact that \eqref{spbeta} is the kernel of the Beta distribution with parameters $1$ and $\eta$ facilitates computation. Note that no publication bias exists when $\eta = 1$ since $\mu(p) = 1$ while $\mu(p)$ is large for small $p$ when $\eta$ is large. \cite{mathur2024sensitivity} showed that the model \eqref{spbeta} can accommodate more complex selection mechanisms by including various of densities of Beta distribution. The bias of the $\widehat{\text{FDR}}$ under the SPM \eqref{spbeta} is provided as
\begin{equation}
\text{Bias}(\widehat{\text{FDR}})=\frac{G(\alpha,\lambda,\pi_0,\gamma,\eta)}{K(\alpha,\pi_0,\gamma,\eta)}, \label{eq:biasbeta}
\end{equation}
where 
\[
\begin{aligned}
G(\alpha,\lambda,\pi_0,\gamma,\eta)=&\frac{\alpha}{\alpha-\lambda}\left[\frac{\pi_0}{\eta} \{ (1-\lambda)^{\eta}-(1-\alpha)^{\eta} \} +(1-\pi_0)\gamma \{ B(\alpha; \gamma, \eta)-B(\lambda; \gamma, \eta) \} \right] \\
&-\frac{\pi_0}{\eta} \{ 1-(1-\alpha)^{\eta} \}, \\
K(\alpha,\pi_0,\gamma,\eta)=&\frac{\pi_0}{\eta}\{1-(1-\alpha)^{\eta}\}+(1-\pi_0)\gamma B(\alpha; \gamma, \eta),
\end{aligned}
\]
and $B(x; a,b)=\int_0^x t^{a-1}(1-t)^{b-1}dt$ is the incomplete Beta function for $\gamma\in(0,1)$ and $\eta\in [1,\infty)$. 

Again, we can easily explore the relationship between the bias and the SPM using this expression. Observe, when $\gamma = 0$ and $\eta = 1$ we recover 0 bias.  This is expected because no selection occurs and consequently $P_i'$ still vary according to a least favorable configuration.  Figure \ref{betaspm} illustrates that bias can again be positive or negative (recall Figure \ref{stepSPM}). Negative bias for some larger $\eta$ is observed, especially when the power is low and $\pi_0$ large. Note that \cite{Johnson2017} suggest $\hat{\pi}_0 = 0.93$ with the estimated power $0.75$ while \cite{Schneck2023} estimates power as low as 0.233. If indeed $\pi_0$ is high and power is low, as suggested by these estimates, the FDR is then at risk of being substantially underestimated under \eqref{spbeta}. For example, for $\pi_0=0.93$, power=$0.233$, $\eta=35$, $\lambda=0.045$ and $\alpha=0.05$, the bias is -0.359.

\newpage
\begin{figure}[H]
    \centering
    \includegraphics[width=1.0\textwidth, height = 15cm]{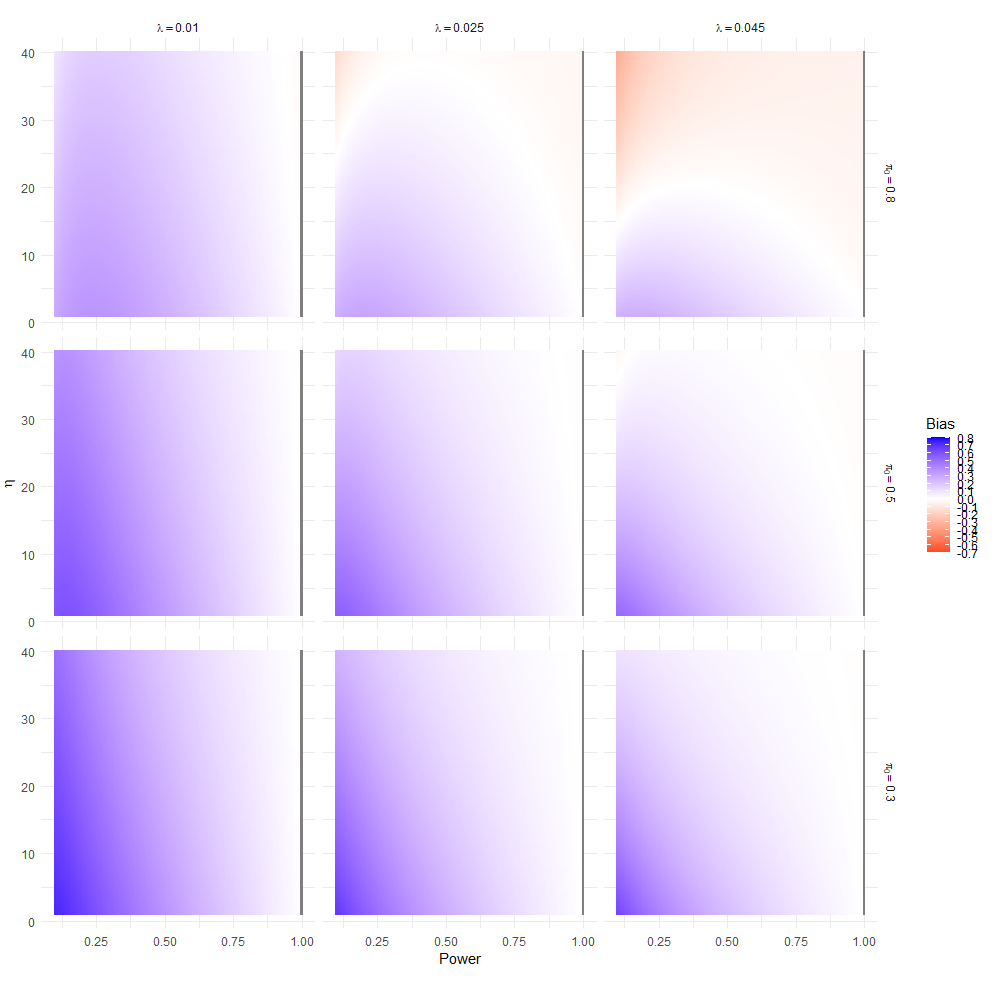}   
    \caption{Bias in \eqref{eq:biasbeta} over $\eta$ and power (corresponding to $\gamma$) at $\lambda = 0.01,0.025,0.045$ and $\pi_{0}=0.3,0.5,0.8$ when $\alpha=0.05$}
    \label{betaspm}
\end{figure}

The exact expression for the variance of $\widehat{\text{FDR}}$ is computationally cumbersome since it involves multinomial probability calculations, which require large factorials when $n$ is large. However, an asymptotic approximation based on the multivariate normal approximation of the multinomial distribution is available and may be useful. It is provided below.  

\begin{theorem}
\label{varMME}
\normalfont
Under models in Definitions \ref{MM}-\ref{def:PPM}, for large enough $n$, the variance of $\widehat{\text{FDR}}$ in Definition \ref{hungmme} can be approximated as
\begin{equation} 
    \text{Var}(\widehat{\text{FDR}})\approx \frac{\alpha^2}{(\alpha-\lambda)^2}\frac{q_1q_2}{n(q_1+q_2)^3}=O(n^{-1}),\label{eq:varMME}
\end{equation}
where $q_1=E_F[\mu(P)]^{-1} E_F[\mu(P)1_{[0,\lambda]}(P)]$, $q_2=E_F[\mu(P)]^{-1}E_F[\mu(P)1_{[\lambda,\alpha]}(P)]$ and $n$ is the total number of published p-values between 0 and 1.
\end{theorem}

\begin{rem}
\normalfont Since the bias depends on $\lambda$ (but not $n$) and the variance is negligible for large $n$ regardless of $\lambda$, it may be reasonable to focus on the bias expression when choosing $\lambda$ when $n$ is large.    
\end{rem}

\section{Application}

To illustrate how the proposed bias and variance expressions can facilitate a sensitivity analysis for FDR estimation, we apply them to the two real datasets studied in \cite{Johnson2017} and \cite{Jager2014}. We begin with \cite{OSC2015} from which \cite{Johnson2017} analyzed a selection of 73 published and corresponding replication p-values in psychology. Under \eqref{johnsonrule}, \cite{Johnson2017} proposed the two estimates of the null proportion: $\hat{\pi}_0=0.930$ and $0.886$, each obtained from a different model for effect sizes. With the estimated power of $0.750$, their FDR estimates were $0.479$ and $0.350$ based on $\hat{\pi}_0=0.930$ and $0.886$, respectively. Table \ref{tablejohnson1} summarizes the resulting $\widehat{\text{FDR}}$ in Definition \ref{hungmme} for each $\lambda$, along with the corresponding bias and standard error (SE) based on the originally published p-values using expressions \eqref{df:biasMME} and \eqref{eq:varMME}. 

First observe when $\lambda=0.010$, $\widehat{\text{FDR}}=0.442$ with bias of approximately 0.1 and SE of 0.073, which closely aligns with the results in \cite{Johnson2017}. For other reasonable values of $\lambda$, say $\lambda$ between the second and third quartile of the observed $p$-values, we also recover a bias correct $FDR$ estimate between 0.370 and 0.651.  

\begin{rem}
\normalfont Observe when $\lambda$ is greater than the 3rd quartile, however, the standard error is large and point estimates can be small.  This may be anticipated here since the number of observed $p$-values is small and the variance is large, especially when $\lambda$ is large. For example, only five of the 73 selected p-values are between 0.025 and 0.045 and none are between 0.045 and 0.052. While results in this manuscript assume $\lambda$ is fixed, in practice it may be reasonable to utilize the median of the selected $p$-values, or some other prespecified quantile, so as to avoid such erratic behavior.     
\end{rem}

\renewcommand{\arraystretch}{1.4}
\begin{table}[t]
    \centering
    \resizebox{\textwidth}{!}{
    \begin{tabular}{|c|c|c|c|c|c|c|c|c|}
    \hline
    & $\lambda$ & $0.004 \ (Q_2)$ & 0.010 & $0.012 \ (Q_3)$ & 0.015 & 0.025 & 0.035 & 0.045 \\
    \cline{2-9}
    & $\widehat{\text{FDR}}\rule{0pt}{2.8ex}$ & 0.495 & 0.442 & 0.281 & 0.181 & 0.138 & 0.044 & 0.000 \\
    \hline
    \multirow{2}{5em}{$\pi_{0}=0.930$} & $\text{Bias}$ & 0.125 & 0.096 & 0.089 & 0.083 & 0.069 & 0.060 & 0.054 \\
    & $\text{SE}$ & 0.064 & 0.073 & 0.077 & 0.082 & 0.103 & 0.138 & 0.227 \\
    \hline
    \multirow{2}{5em}{$\pi_{0}=0.886$} & $\text{Bias}$ & 0.156 & 0.119 & 0.111 & 0.104 & 0.086 & 0.075 & 0.067\\
    & $\text{SE}$ & 0.056 & 0.062 & 0.065 & 0.068 & 0.084 & 0.110 & 0.178 \\
    \hline    
    \end{tabular}
    }
    \\[4pt]
    \caption{\normalfont Table of $\widehat{\text{FDR}}$, bias and standard error computed with the $p$-values analyzed in \cite{Johnson2017} over each $\lambda$ for $\pi_{0}=0.930, 0.886$ under \eqref{johnsonrule}. $Q_{2}$ and $Q_{3}$ denote the second and third quartiles of the original p-values, respectively.}
\label{tablejohnson1}
\end{table}

In many applications calling for FDR estimation with selected $p$-values parameter estimates necessary to utilize bias and variance expressions in this manuscript are not readily available since replicated $p$-values may not be available. Here, $\widehat{\text{FDR}}$ in Definition \ref{hungmme}, expressions \eqref{df:biasMME} and \eqref{eq:varMME} facilitate FDR estimates and sensitivity analysis. Let us illustrate, recall mined p-values from five major medical journal abstracts between 2000 and 2010 in \cite{Jager2014} and reanalyzed in \cite{Schimmack2023}. Recall the FDR estimate in \cite{Jager2014} is $0.188$ for the \textit{Lancet} under the ``$p<0.05$'' SPM, i.e. the SPM in \eqref{cs} with $\rho=1.0$. We compute the $\widehat{\text{FDR}}$ and its corresponding bias with the same published $p$-values under the step function SPM \eqref{cs} and $\rho=1.0,0.8,0.6,0.4$, which is summarized in Table \ref{tableJL} assuming power $=0.8$. While each $\widehat{\text{FDR}}$ in Table \ref{tableJL} is larger than 0.188, when we subtract bias we again recover similar results for most SPMs and $\lambda$ choices. For example, $\widehat{\text{FDR}}-\text{Bias}(\widehat{\text{FDR}})=0.2602-0.0699=0.1903$ when $\lambda=0.035$, $\rho=1$. However, let us suppose that ``$p\leq 0.05$'' rule is not reasonable, but rather $\rho=0.4$. Now the bias is smaller, and even negative for some $\lambda$. Our expression allow us to report that (using $\lambda=0.035$) we may estimate $\widehat{\text{FDR}}-\text{Bias}(\widehat{\text{FDR}})\approx 0.19$ if $p<0.05$ rule reasonably approximate the true selection mechanism.  However if p-values $<2^{-1}\times0.05=0.025$ are $(0.4)^{-1}=2.5$ times more likely to be published than p-values between 0.025 and 0.05, then we may anticipate an FDR of $\widehat{\text{FDR}}-\text{Bias}(\widehat{\text{FDR}}) = 0.2602+0.0350=0.2952$. We suspect policy makers interested in FDR estimation may be more interested in several different FDR estimates alongside summaries of their corresponding SPMs, rather than one estimate under the $p<0.05$ rule only. This may be especially true if such estimates are not computationally complex. 

\renewcommand{\arraystretch}{1.2}
\begin{table}[t]
    \centering
    \resizebox{\textwidth}{!}{ 
    \begin{tabular}{|c|c|c|c|c|c|c|}
    \hline
      & $\lambda$ & $0.0100$ & $0.0150$ & $0.0250$ & $0.0350$ & $0.0450$\\
    \cline{2-7}
     & $\widehat{\text{FDR}}\rule{0pt}{2.8ex}$    & 0.3449  &  0.3430  & 0.2962  & 0.2602 & 0.2718 \\  
    \hline
    \multirow{4}{5.0em}{$\text{Bias}(\widehat{\text{FDR}})$} & $\rho = 1.0$ & 0.1130 & 0.0980 & 0.0805 & 0.0699 & 0.0625 \\
    & $\rho= 0.8$ & 0.1007 & 0.0802 & 0.0457 & 0.0370 & 0.0309 \\
    & $\rho= 0.6$ & 0.0878 & 0.0613 & 0.0088 & 0.0020 & -0.0026 \\
    & $\rho= 0.4$ & 0.0740 & 0.0413 & -0.0303 & -0.0350 & -0.0382 \\
    \hline
    \end{tabular}
    }
    \\[4pt]    
    \caption{\normalfont Table of $\widehat{\text{FDR}}$ and corresponding bias computed with the published $p$-values from \textit{Lancet} analyzed in \cite{Jager2014} over each $\lambda$ and $\rho$ when $\pi_{0}=0.8$ and power $=0.8$ under \eqref{cs}}
    \label{tableJL}
\end{table}

\section{Conclusion}
In this paper, we provided a definition of a selection probability model (SPM) and the resulting post-selective FDR. Closed-form expressions for the bias and variance of an FDR estimator under an arbitrary SPM were derived. The numerical results suggest that the estimator can be negatively biased for large $\pi_0$, especially when p-values much smaller than 0.05 are more likely to be published than $p$-values near 0.05. The expressions were validated in practice using the replication study from \cite{OSC2015}, where parameter estimates and a non-standard SPM from \cite{Johnson2017} were available. In practice, FDR estimators that use only selected $p$-values necessarily rely on  selection model assumptions, which are difficult to verify.  However, expressions provided here readily facilitate sensitivity analysis, so that policy makers and researchers can better gauge how selection model assumptions may impact FDR estimates.  

This manuscript focused on bias and variance expressions for a simple post-selective FDR estimator under generic and specific parametric SPMs. Our focus on a simple estimator allowed for closed-form expressions that facilitate a sensitivity analysis. It would be interesting to develop similar bias and variance expressions for more sophisticated estimators when SPMs may be misspecified, like maximum likelihood estimators in \cite{Jager2014}.

\bibliographystyle{imsart-nameyear} 
\bibliography{replicability}   

\appendix
\renewcommand{\theequation}{A.\arabic{equation}}
\setcounter{equation}{0}
\section{}

\begin{proof}[Proof of Theorem \ref{bayes}]
 It suffices to show that under Definition \ref{def:FDR}, mixture model in Definition \ref{MM} and SPM in Definition \ref{def:PPM} 
 $$\Pr(H_i = 0|\delta_i = 1, P_i\leq \alpha) = E\left[\frac{\sum_{i=1}^{m}(1-H_i)\delta_i 1_{[0,\alpha]}(P_i)}{\sum_{i=1}^{m}\delta_i1_{[0,\alpha]}(P_i)}\bigg\rvert\sum_{i=1}^m\delta_i1_{[0,\alpha]}(P_i)>0\right]$$
The proof is identical to the proof of Theorem 1 in \cite{Storey2003} but replaces $1_{[0,\alpha]}(P_i)$ with $\delta_i 1_{[0,\alpha]}(P_i)$ throughout. Letting $\omega_i = \delta_i 1_{[0,\alpha]}(P_i)$ for readability, 
first note that,
   \begin{align}
        \text{FDR} &= E\left[\frac{V}{R}|R>0\right]\nonumber\\
             &= \sum_{k=1}^m E[\frac{V}{R}|R=k]\Pr(R=k|R>0)\nonumber\\
             &= \sum_{k=1}^m E[\frac{V}{k}|R=k]\Pr(R=k|R>0), \label{b1}
    \end{align}
    where $V=\sum_{i=1}^m (1-H_i)\omega_i$, $R=\sum_{i=1}^m \omega_i$ and $(H_i, \omega_i)$ is an i.i.d. random vector. Thus, observe in (\ref{b1}), $V\sim \text{Binomial}[k, \Pr(H_i=0|\omega_i=1)]$ given that $R=k$. Thus, for the first term of \eqref{b1},

    \begin{align}
        E[V|R=k]&=E[\sum_{i=1}^m (1-H_i)\omega_i|\sum_{i=1}^m \omega_i=k]\nonumber\\
        &=E[\sum_{i=1}^m (1-H_i)\omega_i|\omega_1=\omega_2...=\omega_k=1,\omega_{k+1}=\omega_{k+2}=...=\omega_m=0]\label{b2}\\
        &=E[\sum_{i=1}^k (1-H_i)|\omega_1=\omega_2=...=\omega_k=1,\omega_{k+1}=\omega_{k+2}=...=\omega_m=0]\nonumber\\
        &=E[\sum_{i=1}^k (1-H_i)|\omega_1=\omega_2=...=\omega_k=1]\label{b3}
    \end{align}
    The equal sign in (\ref{b2}) is because $\omega_i$ could only be 0 or 1. So $\sum_{i=1}^m \omega_i=k$ implies that, without loss of generality, $\omega_1$ to $\omega_k$ are 1 and the rest are 0. 

    Then recall that in p-value mixture model in Definition \ref{MM}, $H_i\overset{\text{i.i.d.}}{\sim} \text{Bernoulli}(1-\pi_0)$. Observe that $\omega_i$ is also i.i.d. Bernoulli distributed. Thus (\ref{b3}) could be written
    \begin{align}
        E[\sum_{i=1}^k (1-H_i)|\omega_1=\omega_2=...=\omega_k=1]&=\sum_{i=1}^k E[(1-H_i)|\omega_i=1]\nonumber\\
        &=k\times \Pr(H_i=0|\omega_i=1),\label{b4}
        \end{align}
        for some $i\in\{1,2,...,k\}$. So plugging  (\ref{b4}) into  (\ref{b1}), we have
        \begin{align}
        \text{FDR} &= \sum_{k=1}^m E[\frac{V}{k}|R=k]\Pr(R=k|R>0)\nonumber\\
        &=\sum_{k=1}^m \frac{k\times \Pr(H_i=0|\omega_i=1)}{k}\Pr(R=k|R>0)\nonumber\\
        &=\Pr(H_i=0|\omega_i=1)\sum_{k=1}^m\Pr(R=k|R>0)\nonumber\\
        &=\Pr(H_i=0|\omega_i=1)\label{be}\\
        &=\Pr(H_i=0|\delta_i 1_{[0,\alpha]}(P_i)=1)\nonumber\\
        &=\Pr(H_i=0|\delta_i=1,P_i\leq\alpha).\nonumber
        \end{align}
        The equal sign in (\ref{be}) is because $R>0$ if and only if $R\in\{1,2,...,m\}$.

\end{proof}

\renewcommand{\thesection}{3}
\setcounter{section}{3} 
\setcounter{proposition}{1} 

\begin{proof}[Proof of Theorem \ref{biasMME}] 
First note that the $bias=E[\widehat{\text{FDR}}-\text{FDR}]=E[\widehat{\text{FDR}}]-\text{FDR}$, where $\text{FDR}=\Pr(H_i=0|\delta_i=1, P_i\leq\alpha)$ is a constant provided in Theorem \ref{bayes}. The Bayesian form of the $\text{FDR}$ can be expanded below:
\begin{align}
    \text{FDR}&=\Pr(H_i=0|\delta_i=1,P_i\leq\alpha)\nonumber\\
       &=\frac{\Pr(H_i=0)\Pr(\delta_i=1,P_i\leq\alpha|H_i=0)}{\Pr(\delta_i=1,P_i\leq\alpha)}\nonumber\\
       &=\frac{\Pr(H_i=0)E[\delta_i1_{[0,\alpha]}(P_i)|H_i=0]}{E[\delta_i 1_{[0,\alpha]}(P_i)]}\nonumber\\
        &=\frac{\Pr(H_i=0)E_{F_0}[\delta_i1_{[0,\alpha]}(P_i)]}{E[\delta_i 1_{[0,\alpha]}(P_i)]}\nonumber\\
       &=\frac{\Pr(H_i=0)E_{F_0}[E[\delta_i1_{[0,\alpha]}(P_i)|P_i]]}{E_{F}[E[\delta_i1_{[0,\alpha]}(P_i)|P_i]]}\nonumber\\
        &=\frac{\Pr(H_i=0)E_{F_0}[1_{[0,\alpha]}(P_i)E[\delta_i|P_i]]}{E_{F}[1_{[0,\alpha]}(P_i)E[\delta_i|P_i]]}\label{ci}\\
       &=\frac{\pi_0E_{F_0}[1_{[0,\alpha]}(P)\mu(P)]}{E_{F}[1_{[0,\alpha]}(P)\mu(P)]}. \label{muE}
\end{align}
Equation (\ref{ci}) holds because $1_{[0,\alpha]}(P_i)$ is a constant if $P_i$ is given. Equation (\ref{muE}) holds because $\mu(p)$ is defined to be $\Pr(\delta_i=1|P_i=p)$ which equals $E[\delta_i|P_i=p]$.
Therefore, to calculate the bias, it is sufficient to derive $E[\widehat{\text{FDR}}]$. Recall $\widehat{\text{FDR}}=\sum_{i=1}^{n'}\frac{1_{(\lambda',1]}(\frac{p_i}{\alpha})}{n'(1-\lambda')}$ for $\lambda'\in[0,1]$. Multiplying both the numerator and the denominator by $\alpha$, it is equivalent to $\widehat{\text{FDR}}=\sum_{i=1}^{n'}\frac{\alpha 1_{(\alpha\lambda',\alpha]}(P_i)}{n'(\alpha-\alpha\lambda')}$. Letting $\lambda=\alpha\lambda'\in[0,\alpha]$, we have
\begin{align}
    \widehat{\text{FDR}}&=\frac{\alpha}{\alpha-\lambda}\sum_{i=1}^{n'}\frac{ 1_{(\lambda,\alpha]}(P_i)}{n'}\nonumber\\
    &=\frac{\alpha}{\alpha-\lambda}\frac{\sum_{i=1}^m 1_{(\lambda,\alpha]}(P_i)\delta_i}{\sum_{i=1}^m 1_{[0,\alpha]}(P_i)\delta_i}.\label{nprhl}
\end{align}
for $\lambda\in[0,\alpha]$. Equation (\ref{nprhl}) holds because $n'$ is the number of observed/published p-values that are less than or equal to $\alpha$ and m is the total number of the p-values which is an unknown constant. Under mixture model in Definition \ref{MM} and SPM in Definition \ref{def:PPM}, the numerator and denominator of equation (\ref{nprhl}) are both Binomial random variables where $\sum_{i=1}^m 1_{(\lambda,\alpha]}(P_i)\delta_i\sim \text{ Binomial}(n, \Pr(\lambda<P_i\leq\alpha|\delta_i=1))$ and $\sum_{i=1}^m 1_{[0,\alpha]}(P_i)\delta_i\sim \text{ Binomial}(n, \Pr(P_i\leq\alpha|\delta_i=1))$. Based on this observation, the expectation of the estimator can be written below
\begin{align}
   E[\widehat{\text{FDR}}]&=\frac{\alpha}{\alpha-\lambda}E\left[\frac{\sum_{i=1}^m 1_{(\lambda,\alpha]}(P_i)\delta_i}{\sum_{i=1}^m 1_{[0,\alpha]}(P_i)\delta_i}\right]\nonumber\\
   &=\frac{\alpha}{\alpha-\lambda}E\left[E\left[\frac{\sum_{i=1}^m 1_{(\lambda,\alpha]}(P_i)\delta_i}{\sum_{i=1}^m 1_{[0,\alpha]}(P_i)\delta_i}\bigl\rvert\sum_{i=1}^m 1_{[0,\alpha]}(P_i)\delta_i=k \right] \right]\nonumber\\
   &=\frac{\alpha}{\alpha-\lambda}E\left[\frac{1}{k}E\left[\sum_{i\in K} 1_{(\lambda,\alpha]}(P_i)\delta_i\bigl\rvert\sum_{i=1}^m 1_{[0,\alpha]}(P_i)\delta_i=k \right] \right]\nonumber\\
   &=\frac{\alpha}{\alpha-\lambda}E\left[\frac{1}{k}k \Pr(\lambda<P_i\leq\alpha|P_i\leq\alpha,\delta_i=1) \right]\nonumber\\
   &=\frac{\alpha}{\alpha-\lambda} \Pr(\lambda<P_i\leq\alpha|P_i\leq\alpha,\delta_i=1),\nonumber\\
   &=\frac{\alpha}{\alpha-\lambda} \frac{\Pr(\lambda<P_i\leq\alpha,\delta_i=1)}{\Pr(P_i\leq\alpha, \delta_i=1)},\nonumber\\
   &=\frac{\alpha}{\alpha-\lambda} \frac{\frac{\Pr(\lambda<P_i\leq\alpha,\delta_i=1)}{\Pr(\delta_i=1)}}{\frac{\Pr(P_i\leq\alpha, \delta_i=1)}{\Pr(\delta_i=1)}},\nonumber\\
   &=\frac{\alpha}{\alpha-\lambda} \frac{\Pr(\lambda<P_i\leq\alpha|\delta_i=1)}{\Pr(P_i\leq\alpha| \delta_i=1)},\nonumber\\
   &=\frac{\alpha}{\alpha-\lambda} \frac{E_F[\mu(P)1_{(\lambda,\alpha]}(P)]}{E_F[\mu(P)1_{[0,\alpha]}(P)]},\nonumber
\end{align}
where $K = \{i\in \mathcal{M}: 1_{(0,\alpha]}(P_i)\delta_i=1\}$, for $\lambda\in[0, \alpha]$.
\end{proof}

\begin{proof}[Proof of Proposition \ref{varMME}]
Recall the MME \eqref{nprhl}
\begin{equation}
\widehat{\text{FDR}}=\frac{\alpha}{\alpha-\lambda}\big(1-\frac{X_1}{X_1+X_2}\big), \nonumber
\end{equation}
 where $X_1=\sum 1_{[0,\lambda]}(P_i)\delta_i$, $X_2=\sum 1_{(\lambda,\alpha]}(P_i)\delta_i$ and $X_3=\sum 1_{(\alpha,1]}(P_i)\delta_i$. Therefore $(X_1,X_2,X_3)'\sim \text{Multinomial}(n, (q_1,q_2,q_3)')$, where $n=X_1+X_2+X_3$ is the total number of published p-values, $q_1=\Pr(P_i\leq\lambda|\delta_i=1)=\frac{E_F[\mu(P)1_{[0,\lambda]}(P)]}{E_F[\mu(P)]}$, $q_2=\Pr(\lambda<P_i\leq\alpha|\delta_i=1)=\frac{E_F[\mu(P)1_{[\lambda,\alpha]}(P)]}{E_F[\mu(P)]}$ and $q_3=\Pr(\alpha<P_i\leq1|\delta_i=1)=\frac{E_F[\mu(P)1_{[\alpha,1]}(P)]}{E_F[\mu(P)]}$. Then we are able to calculate the variance:
\begin{align}
\text{Var}(\widehat{\text{FDR}})&=\text{Var}(\frac{\alpha(1-\frac{X_1}{X_1+X_2})}{\alpha-\lambda}) \nonumber\\
&=\frac{\alpha^2}{(\alpha-\lambda)^2}\text{Var}(1-\frac{X_1}{X_1+X_2})\nonumber\\
&=\frac{\alpha^2}{(\alpha-\lambda)^2}\text{Var}(\frac{X_1}{X_1+X_2}).\label{eq:57}
\end{align}
The variance of $\widehat{\text{FDR}}$ cannot be calculated directly by (\ref{eq:57}) due to the high computational cost when $n$ is large. However a normal approximation and multivariate Delta method can approximate the variance. Recall that $(X_1,X_2,X_3)'\sim \text{Multinomial}(n, (q_1,q_2,q_3)')$. Therefore, asymptotically, $\sqrt{n}\left(\left(\frac{X_1}{n}, \frac{X_2}{n}, \frac{X_3}{n}\right)'- \left(q_1, q_2, q_3\right)'\right)\sim N\left((0,0,0)', \boldsymbol{\Sigma}\right)$, where 
\begin{equation}
\boldsymbol{\Sigma}=\begin{pmatrix}
q_1(1-q_1) & -q_1q_2 & -q_1q_3\\
-q_2q_1 & q_2(1-q_2) & -q_2q_3 \\
-q_3q_1 & -q_3q_2    & q_3(1-q_3)
\end{pmatrix}.\nonumber
\end{equation}
Then by the Delta method, 
\begin{equation}
\sqrt{n}\left(\frac{X_1}{X_1+X_2}- \frac{q_1}{q_1+q_2}\right)\sim N\left(0, \frac{q_1q_2}{(q_1+q_2)^3}\right). \label{var:1}
\end{equation}
So plugging (\ref{var:1}) into (\ref{eq:57}) the approximated variance of $\widehat{\text{FDR}}$ when $n$ is large is
\begin{equation} 
\text{Var}(\widehat{\text{FDR}})\approx \frac{\alpha^2}{(\alpha-\lambda)^2}\frac{q_1q_2}{n(q_1+q_2)^3}. \nonumber
\end{equation}       
\end{proof}

\begin{corollary}
\normalfont
Under the SPM \eqref{cs}, the variance of the MME \eqref{nprhl} is 
\begin{equation}
     \text{Var}(\widehat{\text{FDR}})\approx \frac{\alpha^2}{(\alpha-\lambda)^2}\frac{q_1q_2}{n(q_1+q_2)^3},\nonumber
\end{equation}
where $E_F[\mu(P)]=\pi_0(\rho+1)\alpha/2+(1-\pi_0)\{\rho[\alpha^\gamma-(\alpha/2)^\gamma]+(\alpha/2)^\gamma\}$, \\and $q_1=\frac{\pi_0[\alpha/2+\rho(\lambda-\alpha/2)]+(1-\pi_0)\{(\alpha/2)^\gamma+\rho[\lambda^\gamma-(\alpha/2)^\gamma]\}}{E_F[\mu(P)]}$, $q_2=\frac{\pi_0\rho(\alpha-\lambda) + (1-\pi_0)\rho(\alpha^\gamma-\lambda^\gamma)}{E_F[\mu(P)]}$, if $\lambda\geq \alpha/2$;
$q_1=\frac{\pi_0\lambda+(1-\pi_0)\lambda^\gamma}{E_F[\mu(P)]}$, $q_2=\frac{\pi_0(\rho\frac{\alpha}{2}+\frac{\alpha}{2}-\lambda) + (1-\pi_0)[\rho(\alpha^\gamma-\frac{\alpha}{2}^\gamma)+(\frac{\alpha}{2})^\gamma-\lambda^\gamma]}{E_F[\mu(P)]}$, if $\lambda<\alpha/2$. 
\end{corollary}

\begin{corollary}
\normalfont
Under the SPM \eqref{spbeta}, the variance of the MME \eqref{nprhl} is 
\begin{equation}
     \text{Var}(\widehat{\text{FDR}})\approx \frac{\alpha^2}{(\alpha-\lambda)^2}\frac{q_1q_2}{n(q_1+q_2)^3},\nonumber
\end{equation}
where $E_F[\mu(P)]=\frac{\pi_0}{\eta}+(1-\pi_0)\gamma B(1; \gamma, \eta)$,\\ $q_1=\frac{\frac{\pi_0}{\eta}[1-(1-\lambda)^{\eta}]+(1-\pi_0)\gamma B(\lambda; \gamma, \eta)}{E_F[\mu(P)]}$, $q_2=\frac{\frac{\pi_0}{\eta}[(1-\lambda)^{\eta}-(1-\alpha)^{\eta}]+(1-\pi_0)\gamma [B(\alpha; \gamma, \eta)-B(\lambda; \gamma, \eta)]}{E_F[\mu(P)]}$, and $B(x; a,b)=\int_0^x t^{a-1}(1-t)^{b-1}dt$.
\end{corollary}

\end{document}